\documentclass[%
 reprint,
 amsfonts, amsthm, amsmath,amssymb,
 aps,
]{revtex4-2}

\usepackage{graphicx}
\usepackage{dcolumn}
\usepackage{bm,dsfont,color}
\usepackage[T1]{fontenc}
\def\W{\mathscr{W}\!\mathrm{res}}
\def\Tr{\mathrm{Tr}}
\def\PDO{\Psi\mathrm{DO}}
\def\dv{\mathrm{dvol}_g}

\usepackage[scr=boondoxo,scrscaled=1.05]{mathalfa}

\newcommand{\I}{\mathds{1}}

\begin{document}

\preprint{APS/123-QED}

\title{An impediment to torsion from spectral geometry}

\author{Arkadiusz Bochniak}
\email{arkadiusz.bochniak@mpq.mpg.de}
\affiliation{%
Max-Planck-Institut f{\"u}r Quantenoptik,Hans-Kopfermann-Str.~1, Garching, 85748, Germany
}%
\affiliation{%
Munich Center for Quantum Science and Technology, Schellingstraße~4, M{\"u}nchen, 80799, Germany
}%

\author{Ludwik D\k{a}browski}
\email{dabrow@sissa.it}
\affiliation{%
Scuola Internazionale Superiore di Studi Avanzati, Via Bonomea 265, Trieste, 34136, Italy
}%

\author{Andrzej Sitarz}
\email{andrzej.sitarz@uj.edu.pl}
\affiliation{%
Institute of Theoretical Physics, Jagiellonian University, \L{}ojasiewicza 11, Krak\'ow, 30-348, Poland
}%
\author{Pawe\l{} Zalecki}
\email{pawel.zalecki@doctoral.uj.edu.pl}
\affiliation{%
Institute of Theoretical Physics, Jagiellonian University, \L{}ojasiewicza 11, Krak\'ow, 30-348, Poland
}%

\date{\today}

\begin{abstract}
Modifications of standard general relativity that bring torsion into a game have a long-standing history. However, no convincing arguments exist for or against its presence in physically acceptable gravity models. In this Letter, we provide an argument based on spectral geometry (using methods of pseudo-differential calculus) that suggests that the torsion shall be excluded from the consideration. We demonstrate that there is no well-defined functional extending to the torsion-full case of the spectral formulation of the Einstein tensor. 
\end{abstract}

\maketitle


{\it Introduction.--} Riemannian and Lorentzian geometries are the fundamental basis of modern theoretical physics, especially in the context of gravity theories. General relativity is a particular example of how geometric concepts influence our understanding of the world around us, starting from mathematical ideas explaining the expansion of the universe and ending with applications that make our everyday lives easier, such as the global positioning system. Differential operators play an important role in the mathematical description of physical problems. It is well-known that their spectra not only encode several characteristics of the system's dynamics but also possess information about certain properties of the underlying geometric structure. The latter idea is the essence of the famous Mark Kac question: {\it Can one hear the shape of a drum?} \cite{Kac}, which can be referred to as the birth of modern spectral geometry, a branch of mathematics that intends to study the geometric properties of objects by analyzing the spectra of some operators associated with them. Spectral methods are nowadays widely used in modern physics; in particular, quantum-mechanical systems are naturally studied by examining the spectrum of the associated Hamiltonian. Later, a groundbreaking idea by Alain Connes was used in his formulation of noncommutative geometry: an approach to geometry where algebraic structures and spectra of certain operators play an essential role \cite{Connes_book, Chamseddine:1996rwL}. This formulation allows us to go beyond classical spaces, including models that could be referred to as {\it quantum spaces} and have potential applications in other branches of physics, e.g. in particle physics \cite{Connes96} or condensed matter \cite{Bellissard}. Despite the rapid development of new methods and the results obtained within this field, several open questions remained unanswered, even in the purely classical framework.      

The main objects of interest in classical spectral geometry are functionals defined for (pseudo-)differential operators on certain vector bundles. More precisely, let $M$ be an even-dimensional ($n\equiv\dim(M)=2m$) oriented closed Riemannian manifold, take $E\xrightarrow{} M$ to be a (complex) vector bundle over $M$ and consider the algebra $\PDO(E)$ of classical pseudo-differential operators $(\PDO)$ \cite{MR1852334} acting on sections of this bundle. As demonstrated by Mariusz~Wodzicki \cite{Wodzicki1987}, this algebra possesses a unique (up to a multiplicative constant) trace, called the {\it Wodzicki residue},
\begin{equation}
  \W(P)=\int_M\Bigl(\int_{|\xi|=1}\!\!\!\!\!\!\!\!\Tr \, \bigl( \sigma_{-n}(P)(x,\xi) \bigr)\Bigr) d^nx.
\end{equation}
Here, $\Tr$ denotes the trace over $\mathrm{End}(E)$ (at a given point), and $\sigma_{-n}(P)$ is the symbol of order $-n$ of $P\in\PDO(E)$. We recall that for a differential operator of order $d$, 
$P = \sum_{|\alpha|=0}^d (-i)^{|\alpha|}a_{\alpha}(x)\partial_x^\alpha$, its symbol (in multi-index notation) reads $\sum_{|\alpha|=0}^d a_\alpha(x) \xi^\alpha$; this notion can be extended into the full algebra of $\PDO(E)$ -- cf. \cite{Gilkey}.

One of the crucial observations made in the early days of this field was that the Wodzicki residue of powers of the Laplace operator reproduces certain geometric quantities. More precisely, it has been demonstrated that
\begin{equation}
\begin{aligned}
    &\W(f\Delta^{-m})\sim \int_M f\dv,  \\
    &\W(f\Delta^{-m+1})\sim \int_M f R(g) \dv, 
\end{aligned}
\end{equation}
i.e., the volume form and the Ricci scalar can be spectrally reconstructed; cf. \cite{Connes96, Kalau_Walze, Kastler}. The natural question is whether it is possible to encode spectrally (in terms of Wodzicki residue) other geometric quantities like metric, Riemann tensor, Ricci tensor, or Einstein tensor, etc. The positive answer for the metric and the Einstein tensors was recently found in \cite{Dabrowski23}. Spectral methods were also used to examine the torsion in \cite{Dabrowski24} providing an intrinsic method to detect its presence in the Dirac operator. In this Letter, we argue that torsion's presence in gravity models hinders obtaining geometric quantities within spectral geometry. 

{\it Einstein functional with torsion--} The natural proposed way of considering the effect of torsion is to study the Einstein functional. To do so, it is convenient to consider $M$ to be a spin manifold and the Dirac operator $D$ to act on the sections of the spinor bundle. The Einstein functional is a bilinear functional on the space differential forms, 
\begin{equation}
\label{eq:Ein}
   \mathscr{G}(u,w) =\W\bigl(u \{D, w \}D |D|^{-n}\bigr).
\end{equation}
and, for the Dirac operator that comes from the torsion-free Levi-Civita connection, we have:
\begin{equation}
   \mathscr{G}(u,w) \sim \int_M \dv  G^{ab} u_a w_b, 
\end{equation}
where $G^{ab}$ is the Einstein tensor \cite{Dabrowski23}. Such functionals, which have densities represented as evaluations of tensors on the forms, we call {\em tensor type functionals}.

Following this logic, we compute the spectral Einstein functional \eqref{eq:Ein} in the presence of torsion. We have $D=D_0+B$ with $D_0$ being the torsionless Dirac operator constructed from the Levi-Civita connection lifted to the spinor bundle, $D_0=i \gamma^a\nabla_{a}^{(s)}$, and $B=-\frac{i}{8}T_{abc}\gamma^a\gamma^b\gamma^c$ the $3$-form perturbation, with $\gamma^i$ Clifford generators and $T$ the (anti-symmetric part of) torsion. The square of such Dirac operator is a Laplace-type operator $L$, that is, its symbol, at any given point $x$ on the manifold, is given by 
$\sigma(L)=\mathfrak{a}_2+\mathfrak{a}_1+\mathfrak{a}_0$, 
where, in normal coordinates at this point, the respective homogeneous symbols are,
\begin{equation}
	\begin{split}
		\mathfrak{a}_2&= \bigl(\delta_{ab}+\frac{1}{3}R_{acbd}x^c x^d \bigr)\xi_a\xi_b + o({\bf x}^2),\\
		\mathfrak{a}_1&=i(P_{ab}x^b + S_a)\xi_a +o({\bf x}),\\
		\mathfrak{a}_0&=Q+o({\bf 1}),
	\end{split} \label{symbolL}
\end{equation}
where $R_{acbd}$ is the Riemann tensor and $P_{ab}$, $S_a$, and $Q$ are tensors valued in endomorphisms of the vector bundle, which are evaluated at the given point on the manifold. Note that the usual Laplace operator acting on functions corresponds to $P_{ab} =  \frac{2}{3} \mathrm{Ric}_{ab}$, where $\mathrm{Ric}$ the Ricci tensor
and $S_a = Q = 0$. For a second-order differential operator $O$, with a symbol expressed in normal coordinates around $x$, $\sigma(O)=F^{ab}\xi_a\xi_b + i G^a \xi_a +H + o({\bf 1})$, where $F^{ab}=F^{ba}$, $G^a$ and $H$ are  endomorphisms of the fibre at point $x$, we have 
	\begin{equation}
		\begin{aligned}
			\W(OL^{-m})&=\frac{\nu_{n-1}}{24}\int_M \dv\Tr\,\bigl[24H +12 G^aS_a \\[-1mm]
			& +F^{aa}(-12 Q+ 6P_{bb}-2 R  - 3 S_b S_b) \\
			& + 2 F^{ab}(-6 P_{ab}+2 \mathrm{Ric}_{ab}-3S_a S_b)\bigr],
		\end{aligned}
	\end{equation}
where $R$ is the scalar curvature; see \cite{supp} for the detailed derivation. This fact, together with the knowledge of symbols of $L^{-k}$, allows us to compute the density of the Einstein functional \footnote{We recently became aware of a structurally similar result in \cite{hong24}, however exact values of the coefficients presented therein are inconsistent with our result} using the standard techniques of the pseudo-differential calculus (see e.g. \cite{Gilkey}). The resulting expression is of the form $\mathscr{G}_{D_0}(u,w)(x)+\delta\mathscr{G}(u,w)(x)$ with $\mathscr{G}_{D_0}$ being the standard density of the Einstein functional for the torsionless case and 
\begin{equation}
 		\begin{aligned}
 			\delta\mathscr{G}(u,w)(x)=&\frac{\nu_{n-1}}{2} 
 			\Tr\Big\{ iu_aw_{bc}[\gamma^a,\gamma^b]\{\gamma^c, B_0\}\\
            &+u_a w_{b}\Bigl[ \frac{i}{2} [\gamma^a,\gamma^b]\{\gamma^c, B_c\}\\
      &\hspace{-60pt} + (\delta^{ab}B_0-\gamma^a\{\gamma^b,B_0\})(\{\gamma^c,B_0\}\gamma^c-2B_0)\Bigr]\Big\},
 		\end{aligned}
 	\end{equation}
where the expansions in normal coordinates were used: $B=B_0+B_ax^a+o({\bf x})$, $w=(w_a +w_{ab}x^b+w_{abc} x^bx^c)\gamma^a +o({\bf x^2})$, etc. The above expression has, however, an important disadvantage: it cannot represent a well-defined {\it functional of a tensor type} since there is an explicit dependence also on the derivatives of $w$ unless $\Tr\bigl([\gamma^a, \gamma^b]\{\gamma^c, B_0\}\bigr)$ vanishes. This is, however, not possible for the endomorphism $B$ representing a torsion $T$ unless the torsion itself is zero. One can then argue that perhaps there is a potential modification of the spectral Einstein functional that does not include derivatives and reduces to the classical result in the torsionless case. In this Letter, we argue that no such coherent modification is possible.

{\it Modified Einstein functionals--} By power-counting argument for pseudo-differential calculus, we argue that potential modifications of the spectral Einstein functional are of the form $\mathscr{G}'=\sum_{k\in \mathbb{Z}}\bigl(\alpha_k F_1(k) +\beta_k F_2(k)\bigr)$ with (only finitely many nonzero) $\alpha_k, \beta_k\in \mathbb{C}$, and 
\begin{equation}
    \begin{aligned}
        F_1(k)&=\W\bigl(u |D|^k w |D|^{-n-k+2}\bigr),\\
        F_2(k)&=\W\bigl(uD |D|^k w D|D|^{-n-k}\bigr).
    \end{aligned}
\end{equation}
The previously defined Einstein functional is reconstructed with $F_1(0)+F_2(0)$. To start analyzing the problem, we begin with the expression for $F_1$. The main step in this computation is to find the expressions for the symbol of $[|D|^k,w_b \gamma^b]$ -- for details, see \cite{supp}. 
We then analyze the Wodzicki residue of $u |D|^k w |D|^{-n-k+2}$, taking into account only those terms that depend on the derivatives of $w$ or are linear in $w$ but contain contractions of the Riemann tensor. As a result, we get \cite{supp} 
that, up to an irrelevant global constant,
\begin{equation}
\begin{aligned}
 F_1(k)=& \frac{2-n}{24}\Tr(\I)u_aw_a R \\
  &-\frac{k(n+k-2)}{n}\Tr(\I)u_a w_{acc}\\
 &-i\frac{k(k+n-2)}{2n} u_a w_{bc} \Tr\bigl([\gamma^a,\gamma^b]\{\gamma^c,B_0\}\bigr) 
 \\ &+ (\ldots),
\end{aligned}
\end{equation}
where $(\ldots)$ represents all the remaining terms not pertinent to our analysis. Similarly, by analyzing the symbol of $D|D|^k$, we deduce (see \cite{supp}) that 
\begin{equation}
    \begin{aligned}
        F_2(k)=&\Bigl[\frac{n-2}{24}u_aw_aR +\frac{1}{6}u_aw_a G_{ab} \\
        &\hspace{-20pt}+\frac{k(n+k)}{n+2}u_aw_{acc}
        -\frac{4k(n+k)}{n(n+2)}u_a w_{bab}\Bigr] \Tr(\I)\\
        & \hspace{-20pt}+i\Bigl(\frac{k(n+k)}{2n}+\frac{1}{2}\Bigr)u_a w_{bc}\Tr\Bigl([\gamma^a,\gamma^b]\{\gamma^c, B_0\}\Bigr)\\
        &\hspace{-20pt} +(\ldots).
    \end{aligned}
\end{equation}
Let us now consider the modified spectral Einstein functionals $\mathscr{G}'$ and the necessary and sufficient conditions that
the result does not depend on the derivatives of $w$. First, observe that
to eliminate mixed second-order derivatives of $w$ in $\sum_k \beta_k F_2(k)$ we must have $\sum_k \beta_k k(n+k)=0$. Next, we have to independently eliminate the second-order derivatives in $\sum_k \alpha_k F_1(k)$, which leads to $\sum_k \alpha_k k(n+k-2)=0$.

For terms with derivatives of the first order, we easily deduce that $\sum_k \beta_k=0$ or $\Tr\bigl([\gamma^a,\gamma^b]\{\gamma^c, B_0\}\bigr)$ must vanish identically. However, the former is not acceptable because not only would first-order derivatives be eliminated but also the standard Einstein tensor (even in the torsionless case) would have to vanish. Therefore, the only physically acceptable solution is that $\Tr\bigl([\gamma^a,\gamma^b]\{\gamma^c, B_0\}\bigr)$ vanishes. This is impossible for the case where the perturbation is a $3$-form, which is the case with antisymmetric torsion \footnote{We remark that allowed perturbations of the Dirac operator
for which the term vanishes identically include gauge perturbations,
where $B$ is a one-form.}.

{\it Conclusions --} We have demonstrated that spectral geometry imposes a strong obstruction on the presence of torsion in geometric models: for the spectral Einstein functional to be well defined, the torsion $T$ must vanish \footnote{More precisely, this shows that the antisymmetric part of the torsion has to vanish. The vectorial one also must be trivial due to the self-adjointness of the Dirac operator. Cartan torsion can potentially be allowed as being completely transparent in this approach}. Despite the potentially compelling possibilities of analyzing gravity models with torsion, describing the twisting of the reference frame along geodesics, the structural arguments based on spectral geometry provide a strong indication that physically acceptable models should be based on torsionless connections, so the widely used assumption of describing physical systems based on the Levi-Civita connection seems to be, from the spectral perspective, more fundamental than it was supposed so far. We speculate that spectral geometry methods could lead to further structural obstructions for building blocks of modified gravity models. We postpone this intriguing possibility for future research.  

{\it Acknowledgements--} The work of A.B. was partially funded by the Deutsche Forschungsgemeinschaft (DFG, German Research Foundation) under Germany’s Excellence Strategy - EXC-2111 – 390814868, and supported by the University of Warsaw Thematic Research Programme "Quantum Symmetries".  A.B. is also supported by the Alexander von Humboldt Foundation. A.S. is supported by the Polish National Science Centre grant 2020/37/B/ST1/01540. L.D. is affiliated with GNFM–INDAM (Istituto Nazionale di Alta Matematica) and acknowledges that this research is part of the EU Staff Exchange project 101086394 "Operator Algebras That One Can See". 

\bibliography{apssamp}
\end{document}


\preprint{APS/123-QED}

\title{Supplemental Material for ``An impediment to torsion from spectral geometry''}

\author{Arkadiusz Bochniak}
\email{arkadiusz.bochniak@mpq.mpg.de}
\affiliation{%
Max-Planck-Institut f{\"u}r Quantenoptik,Hans-Kopfermann-Str.~1, Garching, 85748, Germany
}%
\affiliation{%
Munich Center for Quantum Science and Technology, Schellingstraße~4, M{\"u}nchen, 80799, Germany
}%

\author{Ludwik D\k{a}browski}
\email{dabrow@sissa.it}
\affiliation{%
Scuola Internazionale Superiore di Studi Avanzati, Via Bonomea 265, Trieste, 34136, Italy
}%

\author{Andrzej Sitarz}
\email{andrzej.sitarz@uj.edu.pl}
\affiliation{%
Institute of Theoretical Physics, Jagiellonian University, \L{}ojasiewicza 11, Krak\'ow, 30-348, Poland
}%

\author{Pawe\l{} Zalecki}
\email{pawel.zalecki@doctoral.uj.edu.pl}
\affiliation{%
Institute of Theoretical Physics, Jagiellonian University, \L{}ojasiewicza 11, Krak\'ow, 30-348, Poland
}%

\date{\today}

\maketitle

In this Supplemental Material, we provide detailed computations and technical 
parts of the proofs, which were briefly sketched in the main text of our Letter. 
First, we provide the exact formula for certain spectral
functionals formed using the Laplace-type operators. We use the result to 
compute Einstein functional for the Dirac operator perturbed by a zero-order
term. Finally, we demonstrate that there is no coherent modification of the Einstein functional that would be fully tensorial and provides the Einstein tensor.

\section{Spectral functionals for Laplace-type operators}

We begin by considering $S$, a module of smooth sections of a vector bundle (with a fiber $V$) over an $n=2m$-dimensional manifold, and $L: S \to S$, a Laplace-type operator on $S$. Our notation is as in \cite{Dabrowski23}, and we use normal coordinates at a given point $x$ of the manifold and expand homogeneous
symbols as polynomials in normal coordinates up to a given order around $x$. The following results extend the previous ones in \cite{Dabrowski23}.
The definition of the Laplace-type operator is already provided in the main text. Here we recall it for completeness:
\begin{definition}\label{defi1}
Let $M$ be a closed Riemannian manifold with a fixed metric and a given vector bundle $V$.  We define $L$ to be a Laplace-type operator 
acting on smooth sections of this bundle, if its symbol, at any given point on the manifold, is given by 
$\sigma(L)=\mathfrak{a}_2+\mathfrak{a}_1+\mathfrak{a}_0$, 
where, in normal coordinates at this point, the respective homogeneous symbols are,
\begin{equation}
	\begin{split}
		\mathfrak{a}_2&= \bigl(\delta_{ab}+\frac{1}{3}R_{acbd}x^c x^d \bigr)\xi_a\xi_b + o({\bf x}^2),\\
		\mathfrak{a}_1&=i(P_{ab}x^b + S_a)\xi_a +o({\bf x}),\\
		\mathfrak{a}_0&=Q+o({\bf 1}),
	\end{split} \label{symbolL}
\end{equation}
where  $R_{acbd}$ is the Riemann tensor and $P_{ab}$, $S_a$, and $Q$ are tensors valued in endomorphisms of the vector bundle, which are
evaluated at the given point on the manifold.
\end{definition}
For the symbol of $L^{-m}$ multiplied by an auxiliary operator $O$, we have
\begin{proposition}
\label{prop:main}
For a second-order differential operator $O$, with a symbol expressed in normal coordinates  
around  a given point $x$ on $M$,
\begin{equation}
    \sigma(O)=F^{ab}\xi_a\xi_b + i G^a \xi_a +H + o({\bf 1}),
\end{equation}
where $F^{ab}=F^{ba}$, $G^a$ and $H$ are  endomorphisms of the fibre $V$ at point $x$,  we have 
	\begin{equation}
		\begin{aligned}
			\W(OL^{-m})\!=\!\frac{\nu_{n-1}}{24}\!\int_M \!\!\dv\Tr\,\bigl[24H &+12 G^aS_a +F^{aa}(-12 Q+ 6P_{bb}-2 R  - 3 S_b S_b) \\
			& + 2 F^{ab}(-6 P_{ab}+2 \mathrm{Ric}_{ab}-3S_a S_b)\bigr],
		\end{aligned}\label{general_wres_uDwD}
	\end{equation}
where $R$ is the scalar curvature and $\mathrm{Ric}$ the Ricci tensor.
\end{proposition}
\begin{proof}
We first compute the symbols of arbitrary powers of $L$ (cf. \cite[Proposition~2.2]{Dabrowski23}). For each $k\in \mathbb{N}$, the symbols of $L^{-k}$ at a given point on the manifold are, $\sigma(L^{-k})=\mathfrak{c}_{2k}+\mathfrak{c}_{2k+1}+\mathfrak{c}_{2k+2} + \cdots$, with homogeneous symbols, 
\begin{equation}
\begin{aligned}
	&\mathfrak{c}_{2k}=\|\xi\|^{-2k-2}\left(\delta_{ab}-\frac{k}{3}R_{acbd}x^c x^d\right)\xi_a\xi_b+ o({\bf x}^2),\\
	&\mathfrak{c}_{2k+1}=-ik\|\xi\|^{-2k-2}(P_{ab}x^b +S_a)\xi_a, \\
	&\mathfrak{c}_{2k+2}=-k Q\|\xi\|^{-2k-2}+k(k+1)\|\xi\|^{-2k-4}\left(P_{ab}-\frac{1}{2}S_aS_b-\frac{1}{3}\mathrm{Ric}_{ab}\right)\xi_a\xi_b +o({\bf 1}).
\end{aligned}	\label{symbolsLk}
\end{equation}
Indeed, let us see that this holds for $k=1$. A simple computation gives the symbols of $L^{-1}$, 
\begin{equation}
\begin{aligned}
&\mathfrak{c}_2=\|\xi\|^{-4}\left(\delta_{ab}-\frac{1}{3}R_{acbd}x^c x^d\right)\xi_a\xi_b+ o({\bf x}^2), \\
&\mathfrak{c}_3=-i\|\xi\|^{-4}(P_{ab}x^b +S_a)\xi_a + o({\bf x}), \\
&\mathfrak{c}_4=-Q\|\xi\|^{-4} +2 \|\xi\|^{-6}\left(P_{ab}-\frac{1}{2}S_a S_b -\frac{1}{3}\mathrm{Ric}_{ab}\right)\xi_a\xi_b +o({\bf 1}),
\end{aligned} \label{prop:symb_c}
\end{equation}
which is a direct consequence of the general formulas for symbols of the inverse operator; cf. \cite{Gilkey}. Iterating this procedure (see \cite[Lemma~A.1]{Dabrowski23}) leads to the claimed result.

The final result comes from a series of computations of the relevant symbols, integrals of the cosphere, and following identities that hold at $x$:
\begin{align}
	& \int_{||\xi||=1} H\mathfrak c_{2m}\ = \ \nu_{n-1}H+o({\bf 1}),\\
	& \int_{||\xi||=1} iG^a\xi_a\mathfrak c_{2m+1}=\frac12\nu_{n-1}G^aS_a+o({\bf 1}),\\
	& \int_{||\xi||=1} G^a\partial_a\mathfrak c_{2m}\ = \ o({\bf 1}),\\
	& \int_{||\xi||=1} F^{ab}\xi_a\xi_b\mathfrak c_{2m+2}\ =\ \frac{1}{24}\nu_{n-1}\Bigl[F^{aa}(-12Q+6P_{bb}-2R-3S_bS_b)\nonumber\\
	& \phantom{xxxxxxxxxxxxxxxxxxx}
	+2F^{ab}(6P_{ab}-2\mathrm{Ric}_{ab}-3S_aS_b)\Bigr]+o({\bf 1}),\\
	&\int_{||\xi||=1} -i F^{ab}\xi_a\partial_b\mathfrak c_{2m+1}\ =\ -\frac{\nu_{n-1}}{2}F^{ab}P_{ab}+o({\bf 1}),\\
	&\int_{||\xi||=1} - F^{ab}\partial_{ab}\mathfrak c_{2m}\ =\  \frac{\nu_{n-1}}{3}F^{ab}\mathrm{Ric}_{ab}+o({\bf 1}).
\end{align}
Adding these partial results gives us the result presented.
\end{proof}
\begin{corollary} \label{lemma:E}
For a $C^\infty(M)$--endomorphism $E:S \to S$ we have:
	\begin{equation}
 \label{eq:withE}
		\W(EL^{-m+1})=\frac{n-2}{24}\nu_{n-1}\int_M \dv \Tr\left[E\left(-12Q+6P_{aa}-2R-3S_aS_a\right)\right].
	\end{equation}
\end{corollary}
\begin{proof}
It is sufficient to take $F^{ab} = E\delta^{ab}$, $G^a=ES^a$ and $H=EQ$ in the previous proposition, obtaining:
	\begin{equation}
		\begin{aligned}
			\W&(ELL^{-m}) =\frac{\nu_{n-1}}{24}\int_M \dv \Tr\,\bigl[24EQ +12 ES^aS_a + E \delta^{aa}(-12Q+6 P_{bb} \\ 
			&  \phantom{xxxxxxxx} -2R  - 3S_bS_b)+ 2E\delta^{ab}(-6P_{ab} + 2 \mathrm{Ric}_{ab}-3S_aS_b)\bigr] \\
			&= \frac{\nu_{n-1}}{24}\int_M \dv\Tr\,\bigl[ (2-n)12 E Q + 2(2-n) ER  \\
			& \phantom{xxxxxxxx} +6 (n-2) E P_{aa} +  3(2-n) E S_a S_a \bigr] \\
			&= (n-2)\frac{\nu_{n-1}}{24} \int_M \dv \Tr\,\bigl[ -12 E Q -2 ER   + 6  E P_{aa} - 3 E S_a S_a \bigr].
		\end{aligned}
	\end{equation}
\end{proof}
\begin{corollary} \label{cor4}
	For an operator $O$ with a symbol $\sigma(O)=F^{ab}\xi_a\xi_b +i G^a \xi_a +H$, we have,
	\begin{equation}
		\begin{aligned}
			\W \Biggl( \Bigl(O - \frac{1}{n-2}F^{aa}L\Bigr) L^{-m} \Biggr) &= 	
			\frac{\nu_{n-1}}{24}\int_M \dv \Tr[24H +12 G^aS_a \\
			&\phantom{xxxx}+ 2F^{ab}(-6P_{ab}+2 \mathrm{Ric}_{ab}-3S_aS_b)].
		\end{aligned}
	\end{equation}
\end{corollary}
\section{Einstein Functional for Dirac-type operators}
Our goal here is to compute the Einstein functional, defined in \cite{Dabrowski23} as a bilinear functional on the space differential forms, 
\begin{equation}
\label{eq:Ein}
   \mathscr{G}(u,w) =\W\bigl(\hat{u} \{D, \hat{w} \}D |D|^{-n}\bigr),
\end{equation}
for the geometries with (antisymmetric) torsion. Here, $\hat{u}$ denotes the Clifford multiplication by the one-form $u$. We begin with the following.
\begin{proposition}
 \label{prop:Ein}
 For $D\!=\!D_0+B$ with $D_0$ being the torsionless Dirac operator constructed from the Levi-Civita connection lifted to the spinor bundle, $D_0=i \gamma^a\nabla_{a}^{(s)}$, the Einstein functional density 
 in normal coordinates reads:
 \begin{equation}
 		\begin{aligned}
 			\mathscr{G}_D(u,w)(x)=\ &\mathscr{G}_{D_0}(u,w)(x)
 			\\
 			+&\frac{\nu_{n-1}}{2} 
 			\Tr\Big\{ iu_aw_{bc}[\gamma^a,\gamma^b]\{\gamma^c, B_0\}+\frac{i}{2}u_aw_{b}[\gamma^a,\gamma^b]\{\gamma^c, B_c\}\\
 			&\hspace{20pt}+u_aw_b\left[(\delta^{ab}B_0-\gamma^a\{\gamma^b,B_0\})(\{\gamma^c,B_0\}\gamma^c-2B_0)\right]\Big\}. \label{Einstein_B}
 		\end{aligned}
 	\end{equation}
In particular, for $B=-\frac{i}{8}T_{jkl}\gamma^j \gamma^k \gamma^l$, the above correction to $\mathcal{G}_{D_0}$ reads
\begin{equation}
		\begin{aligned}
			3 \cdot 2^{m-1} \nu_{n-1}\Bigl[-u_aw_{bc}T^{0}_{abc}+\frac{1}{8}u_aw_b\Bigl(\delta^{ab}T^{0}_{ijk}T^{0}_{ijk}-4T^{c}_{abc}-6 T^{0}_{ajk}T^{0}_{bjk}\Bigr)\Bigr].
		\end{aligned}
	\end{equation}
 \end{proposition}
\begin{proof}
    Let us take an operator $O=\hat{v}D\hat{w}D$, with
    \begin{equation}
        \hat{u}=u_a\gamma^a +o({\bf 1}), \qquad \hat{w}=(w_b +w_{bc}x^c)\gamma^b +o({\bf x}),
    \end{equation}
Its symbol is $\sigma(O)=F^{ab}\xi_a\xi_b + G^a \xi_a +H$ with
    \begin{equation}
        \begin{split}
            &F^{ab}\xi_a\xi_b=F^{ab}_0\xi_a\xi_b=2 u_c w_b\delta_{ad}\gamma^c\gamma^d \xi_a\xi_b - u_c w_d \gamma^c\gamma^d \delta_{ab}\xi_a\xi_b+o({\bf 1}),\\
            &G^a=G^a_0+iu_c w_b\gamma^c\bigl(2\delta^{ba}B_0+\{\gamma^b,B_0\}\gamma^a-\gamma^b\{\gamma^a,B_0\}\bigr) + o({\bf 1}),\\
            &H=H_0+u_a\gamma^a(i\gamma^c\gamma^bw_{bc}B_0+i\gamma^c\gamma^bw_bB_c+w_bB_0\gamma^bB_0)+o({\bf 1}),
        \end{split}
    \end{equation}
where $G_0^a=-u_c w_{bd}\gamma^c\gamma^d \gamma^b \gamma^a$. Since $O=O_0+\hat{u}D_0 \hat{w}B + \hat{u}B\hat{w}D_0+\hat{u}B\hat{w}B$, with $O_0=\hat{u}D_0\hat{w}D_0$, it remains to express the latter three summands in normal coordinates up to an appropriate order. We have
\begin{equation}
    \begin{split}
        \hat{u}D_0\hat{w}B&=iu_a\gamma^a\gamma^c\gamma^b(w_bB_0\partial_c+w_{bc}B_0+w_bB_c)+ o({\bf 1})\\
        \hat{u}B\hat{w}D_0&=iu_aw_b\gamma^aB_0\gamma^b\gamma^c\partial_c+ o({\bf 1})\\
        \hat{u}B\hat{w}B&=u_aw_b\gamma^aB_0\gamma^bB_0 + o({\bf 1}).
    \end{split}
\end{equation}
Therefore, $F^{ab}\xi_a\xi_b=u_cw_d\gamma^c\gamma^a\gamma^d\gamma^b\xi_a\xi_b+ o({\bf 1})  =2u_cw_b\delta_{ad}\gamma^c\gamma^d\xi_a\xi_b-u_cw_d\gamma^c\gamma^d\delta_{ab}\xi_a\xi_b+ o({\bf 1})$, and, as an immediate consequence, $F_0^{aa}=(2-n)u_cw_d\gamma^c\gamma^d$. The rest follows from a straightforward computation,
applying  Corollary \ref{cor4}.  The difference between functionals is,
\begin{equation}
\begin{aligned}
&	\mathscr{G}_D(u,w)-\mathscr{G}_{D_0}(u,w) = \\
&=\frac{\nu_{n-1}}{24}\int\hbox{Tr\ }\Big(24u_a\gamma^a(i\gamma^c\gamma^bw_{bc}B_0 \!+i\!\gamma^c\gamma^bw_bB_c+w_bB_0\gamma^bB_0)\\
  &  \;\; +12 i \bigl[-u_aw_{bd}\gamma^a\gamma^d\gamma^b\gamma^c
  \!+\! iu_aw_b\gamma^a(2\delta^{bc}B_0
  \!+\! \{\gamma^b,B_0\}\gamma^c \!-\! \gamma^b\{\gamma^c,B_0\})\bigr]\{\gamma^c,B_0\}\\
 &  \;\;  +6(u_cw_b\gamma^c\gamma^a \!+\! u_cw_a\gamma^c\gamma^b
 \!-\! u_cw_d\gamma^c\gamma^d\delta_{ab})\bigl[-2i\{\gamma^a,B_b\} \!+\! \{\gamma^a,B_0\}\{\gamma^b,B_0\}\bigr]\Big).\label{Einstein_B1}
\end{aligned}
\end{equation}
To simplify the above expression, we first collect all the terms with $w_{bc}$:
\begin{equation}
\begin{aligned}
  &\frac{i}{2}u_aw_{bc}\Tr\,[2\gamma^a\gamma^c\gamma^bB_0-\gamma^a\gamma^c\gamma^b\gamma^d\{\gamma^d,B_0\}]\\
&=\frac{i}{2}u_aw_{bc}\hbox{Tr\ }[(2-n)\gamma^a\gamma^c\gamma^bB_0-\gamma^a\gamma^c\gamma^b\gamma^dB_0\gamma^d]
\\
&=\frac{i}{2}u_aw_{bc}\Tr\,[(2-n)\gamma^a\gamma^c\gamma^bB_0-2\delta^{ad}\gamma^c\gamma^b\gamma^dB_0+2\delta^{cd}\gamma^a\gamma^b\gamma^dB_0\\
&\hspace{50pt}-2\delta^{bd}\gamma^a\gamma^c\gamma^dB_0+n\gamma^a\gamma^c\gamma^bB_0]
\\
&
=iu_aw_{bc}\Tr\,[-\gamma^c\gamma^b\gamma^aB_0+\gamma^a\gamma^b\gamma^cB_0]\\
&=iu_aw_{bc}\hbox{Tr\ }[-2\delta^{ab}\gamma^cB_0+\gamma^a\gamma^bB_0\gamma^c+\gamma^a\gamma^b\gamma^cB_0]
\\
&
=iu_aw_{bc}\Tr\,[-2\delta^{ab}\gamma^cB_0+\gamma^a\gamma^b\{\gamma^c,B_0\}].
    \end{aligned}
\end{equation}
Similarly, terms linear in $B$ can be collected into
\begin{equation}
    \begin{aligned}
&\frac{i}{2}u_aw_{b}\Tr\,[\gamma^a(2\gamma^c\gamma^bB_c-\gamma^c\{\gamma^c,B_b\}-\gamma^c\{\gamma^b,B_c\}+\gamma^b\{\gamma^c,B_c\})]\\
&=\frac{i}{2}u_aw_{b}\Tr\,[\gamma^a(\gamma^c\gamma^bB_c-nB_b-\gamma^cB_b\gamma^c-\gamma^cB_c\gamma^b+\gamma^b\gamma^cB_c+\gamma^bB_c\gamma^c)]\\
&=
\frac{i}{2}u_aw_{b}\Tr\,[\gamma^a(\gamma^c\gamma^b+\gamma^b\gamma^c)B_c-n\gamma^aB_b-2\delta^{ac}\gamma^cB_b+n\gamma^aB_b\\
&\hspace{50pt}-\gamma^b\gamma^a\gamma^cB_c+\gamma^a\gamma^bB_c\gamma^c]\\
&=
\frac{i}{2}u_aw_{b}\Tr\,[2\delta^{bc}\gamma^aB_c-2\gamma^aB_b-2\delta^{ab}\gamma^cB_c+\gamma^a\gamma^b(\gamma^cB_c+B_c\gamma^c)]\\
&=
\frac{i}{2}u_aw_{b}\Tr\,[-2\delta^{ab}\gamma^cB_c+\gamma^a\gamma^b\{\gamma^c,B_c\}],
    \end{aligned}
\end{equation}
and finally, terms quadratic in $B$ are
\begin{equation}
    \begin{aligned}
        &\frac{1}{4}u_aw_b\Tr \Bigl[\gamma^a\Bigl(4B_0\gamma^bB_0-4B_0\{\gamma^b,B_0\}-2\{\gamma^b,B_0\}\gamma^c\{\gamma^c,B_0\}+2\gamma^b\{\gamma^c,B_0\}^2\\
        &\hspace{50pt}+\gamma^c\{\gamma^c,B_0\}\{\gamma^b,B_0\}+\gamma^c\{\gamma^b,B_0\}\{\gamma^c,B_0\}-\gamma^b\{\gamma^c,B_0\}^2\Bigr)\Bigr]\\
        &=
\frac{1}{4}u_aw_b\Tr \Bigl[\gamma^a\Bigl(-4B_0^2\gamma^b-2\{\gamma^b,B_0\}\gamma^c\{\gamma^c,B_0\}+\gamma^b\{\gamma^c,B_0\}^2\\
&\hspace{50pt}+\gamma^c\{\gamma^c,B_0\}\{\gamma^b,B_0\}+\gamma^c\{\gamma^b,B_0\}\{\gamma^c,B_0\}\Bigr)\Bigr].
    \end{aligned}
\end{equation}
To simplify the last expression, we expand all the anticommutators ($17$ terms), group them according to the number of gamma operators involved, and reduce the expressions by making use of the algebra of gamma matrices. This allows to write the above expression in the following form
\begin{equation}
\label{eq:fin}
    \begin{split}
        \frac{1}{4}u_aw_b\Tr\Bigl[&\gamma^a\Bigl(-2B_0\gamma^b\gamma^cB_0\gamma^c+(2-n)B_0\gamma^bB_0-\gamma^bB_0\gamma^cB_0\gamma^c+\gamma^cB_0\gamma^cB_0\gamma^b\\
        &+\gamma^cB_0\gamma^bB_0\gamma^c+(n-2)(B_0^2\gamma^b-\gamma^bB_0^2)\Bigr)\Bigr].
    \end{split}
\end{equation}
Next, we observe that
\begin{equation}
    \begin{split}
  &-2B_0\gamma^b\gamma^cB_0\gamma^c-\gamma^bB_0\gamma^cB_0\gamma^c=-2\{\gamma^b,B_0\}\gamma^cB_0\gamma^c+\gamma^bB_0\gamma^cB_0\gamma^c\\
  &=-2\{\gamma^b,B_0\}\{\gamma^c,B_0\}\gamma^c+2n\{\gamma^b,B_0\}B_0+\gamma^bB_0\gamma^cB_0\gamma^c,
    \end{split}
\end{equation}
\begin{equation}
\begin{split}
    \Tr\Bigl[\gamma^a\gamma^cB_0\gamma^bB_0\gamma^c\Bigr]&=
\Tr\Bigl[2\delta^{ac}\gamma^cB_0\gamma^bB_0-\gamma^a\gamma^c\gamma^cB_0\gamma^bB_0\Bigr]\\
&=(2-n)\Tr\Bigl(\gamma^aB_0\gamma^bB_0\Bigr),
\end{split}
\end{equation}
and
\begin{equation}
    \begin{split}
        \Tr\Bigl[\gamma^a\gamma^cB_0\gamma^cB_0\gamma^b\Bigr]
        &=
\Tr\Bigl[2\delta^{ac}\gamma^bB_0\gamma^cB_0- 2\delta^{bc}\gamma^aB_0\gamma^cB_0+ \gamma^c\gamma^b\gamma^aB_0\gamma^cB_0\Bigr]\\
&=\Tr\Bigl[2\gamma^bB_0\gamma^aB_0- 2\gamma^aB_0\gamma^bB_0+ \gamma^b\gamma^aB_0\gamma^cB_0\gamma^c\Bigr]\\
&=\Tr\Bigl[\gamma^b\gamma^aB_0\gamma^c B_0\gamma^c\Bigr].
\end{split}
\end{equation}
Therefore, the expression in Eq.~\eqref{eq:fin} can be written as
\begin{equation}
    \begin{aligned}
&
\frac{1}{4}u_aw_b\Tr\Bigl[\gamma^a\Bigl(-2\{\gamma^b,B_0\}\{\gamma^c,B_0\}\gamma^c+2n\{\gamma^b,B_0\}B_0+2(2-n)B_0\gamma^bB_0\\
&\hspace{50pt}+(n-2)(B_0^2\gamma^b-\gamma^bB_0^2)\Bigr)+(\gamma^b\gamma^a+\gamma^a\gamma^b)B_0\gamma^cB_0\gamma^c\Bigr]
\\
&
=\frac{1}{4}u_aw_b\Tr\Bigl[2\delta^{ab}B_0\{\gamma^c,B_0\}\gamma^c-2n\delta^{ab}B_0^2+\gamma^a(-2\{\gamma^b,B_0\}\{\gamma^c,B_0\}\gamma^c+
\\
&
\hspace{50pt}+2n\{\gamma^b,B_0\}B_0+(4-2n)\{\gamma^b,B_0\}B_0+2(n-2)\gamma^bB_0^2\\
&\hspace{50pt}+(n-2)B_0^2\gamma^b-(n-2)\gamma^bB_0^2)\Bigr]=
\\
&
=\frac{1}{4}u_aw_b\Tr\Bigl[2\delta^{ab}B_0\{\gamma^c,B_0\}\gamma^c-2\gamma^a\{\gamma^b,B_0\}\{\gamma^c,B_0\}\gamma^c+4\gamma^a\{\gamma^b,B_0\}B_0\\
&\hspace{50pt}-2n\delta^{ab}B_0^2+(n-2)(\gamma^a\gamma^b+\gamma^b\gamma^a)B_0^2\Bigr]
\\
&
=\frac{1}{2}u_aw_b\Tr\Bigl[\delta^{ab}B_0\{\gamma^c,B_0\}\gamma^c-\gamma^a\{\gamma^b,B_0\}\{\gamma^c,B_0\}\gamma^c\\
&\hspace{50pt}+2\gamma^a\{\gamma^b,B_0\}B_0-2\delta^{ab}B_0^2\Bigr]
\\
&
=\frac{1}{2}u_aw_b\hbox{Tr\ }[(\delta^{ab}B_0-\gamma^a\{\gamma^b,B_0\})(\{\gamma^c,B_0\}\gamma^c-2B_0)].
    \end{aligned}
\end{equation}
Finally, notice that
\begin{equation}
    \Tr(-2\delta^{ab}\gamma^c B)=-\frac{1}{2}\Tr\bigl(\{\gamma^a,\gamma^b\}\{\gamma^c, B\}\bigr).
\end{equation}
Combining all the above terms concludes the proof.
\end{proof}

\section{Tensorial Functionals on Differential Forms}
As the Einstein functional computed in Proposition \ref{prop:Ein} depends explicitly on the derivatives of the one-form $w$, we investigate here a natural question, 
whether there exist functionals for a perturbed Dirac operator, whose densities are represented as evaluations of tensors on the forms. We refer to them as {\it tensor type functionals}, or {\it tensorial functionals}. We shall verify that
no form of functional, which depends polynomially on $D, |D|, |D|^{-1}$ can satisfy
this condition. Hence, we obtain a consistency condition on possible perturbations
of the Dirac operator to provide a tensorial Einstein functional.

In what follows, we concentrate on the spin Dirac operator $D_0=i \gamma^a\nabla_{e_a}^{(s)}=i\gamma^a e_a -\frac{1}{4}\omega_{abc}\gamma^a\gamma^b\gamma^c$, with the canonical spin connection $\omega_{abc}=\frac{1}{2}(c_{abc}+c_{cab}+c_{cba})$ defined by structure constants $[e_a,e_b]=c_{abc}e_c$. Since in the normal coordinates $c_{pqr}=\frac{1}{2}R_{pqnr}x^n +o({\bf x})$, $\omega_{ijk}=-\frac{1}{2}R_{nijk}x^n +o({\bf x})$ and $\omega_{ijk}\gamma^i\gamma^j\gamma^k=\mathrm{Ric}_{ab}\gamma^b x^a$, we have
\begin{equation}
	D_0=i\gamma^a\Bigl(\partial_a - \frac{1}{6}R_{abcd}x^b x^c \partial_d-\frac{1}{4}\mathrm{Ric}_{ab} x^b\Bigr).
	\label{Dirac}
\end{equation}

\subsection{$F_1(k)$}
The goal in this section is to compute the relevant terms in $F_1(k)=\W\bigl(u |D|^k w |D|^{-n-k+2}\bigr)$. We begin with the following lemma (cf. also \cite[Chapter~25]{Taylor}).
\begin{lemma}\label{lemma:D}
The symbol of $|D|$, $\sigma(|D|)=\mathfrak{d}_1+\mathfrak{d}_{0}+\mathfrak{d}_{-1}+\ldots$, reads
\begin{align}
    \mathfrak{d}_1&=\frac{1}{|\xi|}\Bigl(\delta_{a,b}+\frac{1}{6}R_{acbd}x^c x^d\Bigr)\xi_a\xi_b +o({\bf x^2}), \label{lemII1d1}\\
    \mathfrak{d}_0&=i\frac{\xi_a}{|\xi|}x^b \Bigl(\frac{1}{3}\mathrm{Ric}_{ab} + \frac{1}{8}R_{abjk}\gamma^j\gamma^k\Bigr)-\frac{1}{2|\xi|}\xi_a\{\gamma^a, B\} + o({\bf x}),\\
    \mathfrak{d}_{-1}&=\frac{R}{8|\xi|}-\frac{\mathrm{Ric}_{ab}}{12|\xi|^3}\xi_a\xi_b +\frac{1}{2|\xi|} (i\gamma^a B_a +B_0^2)\nonumber \\
    &-\frac{1}{4|\xi|^3}\xi_a\xi_b \Bigl(i\{\gamma^a, B_b\}+\frac{1}{2}\{\gamma^a, B\}\{\gamma^b, B\}\Bigr) +o({\bf 1}).
\end{align}
\end{lemma}

\begin{proof}
        We start with the observation that 
        \begin{align}
            |D|^2=D^2=&-\left(\delta_{ab}+\frac{1}{3} R_{acbd}x^cx^d+o({\bf x^2})\right)\partial_a\partial_b\nonumber\\
            &+\left(\frac{2}{3} \mathrm{Ric}_{ab}x^b+\frac{1}{4}R_{abjk}\gamma^j\gamma^k x^b+i\{\gamma^a,B\}+o({\bf x})\right)\partial_a\nonumber\\
            &+\left(\frac{1}{4}R+i\gamma^a B_a+B_0^2+o({\bf 1})\right),
        \end{align}
        so $D^2$ is a Laplace-type operator $L$ with 
        \begin{align}
            S_a=i\{\gamma^a,B_0\},\quad& P_{ab}=\frac{2}{3} \mathrm{Ric}_{ab}+\frac{1}{4}R_{abjk}\gamma^j\gamma^k+i\{\gamma^a,B_b\}, \nonumber \\
            Q=\frac{1}{4}R&+i\gamma^a B_a+B_0^2.
        \end{align}
        $|D|$ is a positive first-order pseudodifferential operator with the principal symbol being a scalar. In a flat space, its principal symbol should be just $\mathfrak{d}_1=|\xi|$. In normal coordinates it gives $\mathfrak{d}_1=|\xi|+d^{11}_cx^c+d^{12}_{cd}x^cx^d+o({\bf x^2})$. Eqn.~\eqref{lemII1d1} is obtained by setting $(\mathfrak{d}_1)^2=a_2$ from \eqref{symbolL} and comparing coefficients with the same power of $x$.

        Next, we get $\mathfrak{a}_1=\mathfrak{d}_1 \mathfrak{d}_0+\mathfrak{d}_0 \mathfrak{d}_1-i\partial_\xi^a\mathfrak{d}_1 \partial_a\mathfrak{d}_1=2\mathfrak{d}_1 \mathfrak{d}_0 +o({\bf x})$,
since $\mathfrak{d}_1$ is a scalar and commutes with $\mathfrak{d}_0$ and $-i\partial_\xi^c\mathfrak{d}_1 \partial_c\mathfrak{d}_1=-i\frac{1}{3}|\xi|^{-1}\xi_a\xi_b\xi_cR_{abcd}x^d+o({\bf x})=0+o({\bf x})$, due to (anti-)symmetry of $R_{abcd}$.
        We can write
        \begin{align}
        \mathfrak d_0=\frac{\mathfrak{a}_1}{2\mathfrak{d}_1}+o({\bf x})=\frac{i\xi_a}{2|\xi|}(P_{ab}x^b+S_a)+o({\bf x}).
        \end{align}
        Using the exact form of $S_a, P_{ab}$, we get the result. Similarly,
    \begin{align}
        \mathfrak a_0=\mathfrak{d}_1 \mathfrak{d}_{-1}+(\mathfrak{d}_0)^2+\mathfrak{d}_{-1} \mathfrak{d}_{1}-i\partial_\xi^a\mathfrak{d}_1 \partial_a\mathfrak{d}_0-i\partial_\xi^a\mathfrak{d}_0 \partial_a\mathfrak{d}_1-\frac{1}{2}\partial_\xi^{ab}\mathfrak{d}_1 \partial_{ab}\mathfrak{d}_1,
    \end{align}
        and
        \begin{align}            
        \mathfrak{d}_{-1}=&\frac{1}{2\mathfrak{d}_1}\left(\mathfrak{a}_0- (\mathfrak d_0)^2+i\partial_\xi^a\mathfrak{d}_1 \partial_a\mathfrak{d}_0+i\partial_\xi^a\mathfrak{d}_0 \partial_a\mathfrak{d}_1+\frac{1}{2}\partial_\xi^{ab}\mathfrak{d}_1 \partial_{ab}\mathfrak{d}_1\right)\nonumber \\
          =&\frac{Q}{2|\xi|}+\frac{\xi_a\xi_b}{4|\xi|^3}\left(\frac{1}{2}S_aS_b-P_{ab}+\frac{1}{3}\mathrm{Ric}_{ab}\right)+o({\bf 1}).
        \end{align}
        After inserting $S_a, P_{ab},Q$ and using $R_{abcd}\xi_a\xi_b=0$, we get the stated result.
\end{proof}

As a direct consequence of the above result and \cite[Lemma~A.1]{Dabrowski23}, we get
\begin{lemma}
\label{symbolsDk}
         The homogeneous symbols $\sigma(|D|^k)=\mathfrak{d}_k+\mathfrak{d}_{k-1}+\mathfrak{d}_{k-2}+\ldots$ read
     \begin{align}
         \mathfrak{d}_k&=\mathfrak{d}^0_k,\\
         \mathfrak{d}_{k-1}&=\mathfrak{d}^0_{k-1} -\frac{k}{2}|\xi|^{k-2} \xi_a\{\gamma^a, B\} +o({\bf x}),
         \\
         \mathfrak{d}_{k-2}&=\mathfrak{d}^0_{k-2}+\frac{k}{2}|\xi|^{k-2}(i\gamma^a B_a +B_0^2)\nonumber \\
         &+\frac{k(k-2)}{4}|\xi|^{k-4} \xi_a\xi_b \Bigl(i\{\gamma^a, B_b\}+\frac{1}{2}\{\gamma^a, B_0\}\{\gamma^b, B_0\}\Bigr)+o({\bf 1}),
     \end{align}
     with
          \begin{align}
         \mathfrak{d}_k^0&=|\xi|^{k-2}\Bigl(\delta_{a,b}+\frac{k}{6}R_{acbd}x^c x^d\Bigr)\xi_a\xi_b +o({\bf x^2}),\\
         \mathfrak{d}_{k-1}^0&=ik|\xi|^{k-2}\xi_a x^b \Bigl(\frac{1}{3}\mathrm{Ric}_{ab} +\frac{1}{8}R_{abjk}\gamma^j\gamma^k\Bigr) +o({\bf x}),\\
         \mathfrak{d}^0_{k-2}&=\frac{k}{8}|\xi|^{k-2} R +\frac{k(k-2)}{12}|\xi|^{k-4}\mathrm{Ric}_{ab}\xi_a\xi_b + o({\bf 1}),
     \end{align}
     the homogeneous symbols of $\sigma(|D_0|^k)=\mathfrak{d}_k^0+\mathfrak{d}^0_{k-1}+\mathfrak{d}^0_{k-2}+\ldots$.
\end{lemma}

\begin{remark}
    Note that Corollary~A.2 in \cite{Dabrowski23} was derived only for $k\in \mathbb N$, but using it with the $k$-th power of $P^{-1}$ gives the same result as direct use of $-k$, so this Corollary and, consequently, both Eq.~\eqref{symbolsLk} and Lemma~\ref{symbolsDk} hold also for negative integers $k$. Moreover, direct change $k\to -\frac{k}{2}$ in Eq.~\eqref{symbolsLk} leads to the same results in Lemma~\ref{symbolsDk} as following the steps of the proof of Lemma~\ref{lemma:D} and \cite[Corollary~A.2]{Dabrowski23}. Consequently, Proposition~\ref{prop:main} and the Corollaries that follow it are also valid for odd $n$, thus $m$ being half-integers. Therefore, the spectral functionals of the form as in Proposition~\ref{prop:main} defined previously for the even-dimensional case have the same form for odd-dimensional manifolds up to the fiber's dimension. 
\end{remark}

By a straightforward application of the pseudo-differential operator techniques \cite{Gilkey}, we also have the following.
\begin{lemma}\label{lemm:WD40}
     The homogeneous symbols $\sigma\bigl([|D|^k,w_b \gamma^b]\bigr)=\mathfrak{w}_k+\mathfrak{w}_{k-1}+\mathfrak{w}_{k-2}+\ldots$ read
     \begin{align}
         \mathfrak{w}_k&=o({\bf x^2}),\\
         \mathfrak{w}_{k-1}&=\mathfrak{w}_{k-1}^0-\frac{k}{2}|\xi|^{k-2} w_b\bigl[\{\gamma^c, B\},\gamma^b\bigr]+o({\bf x}),
         \\
         \mathfrak{w}_{k-2}&=\mathfrak{w}^0_{k-2}+\frac{k}{2}|\xi|^{k-2}w_b\bigl[i\gamma^c B_c +B_0^2, \gamma^b\bigr] \nonumber \\
         &+\frac{k(k-2)}{4}|\xi|^{k-4} \xi_c \xi_d w_b \Bigl[i\{\gamma^c, B_d\}+\frac{1}{2}\{\gamma^c, B_0\}\{\gamma^d, B_0\},\gamma^b\Bigr]\nonumber \\
         &+i\frac{k}{2}|\xi|^{k-2} \{\gamma^c, B_0\}w_{bc}\gamma^b + i\frac{k(k-2)}{2}|\xi|^{k-4}\xi_c\xi_d \{\gamma^c, B_0\}w_{bd}\gamma^b + o({\bf 1}),
     \end{align}
     with
     \begin{align}
         \mathfrak{w}^0_{k-1}&=i\frac{k}{2}|\xi|^{k-2}\xi_c w_b x^d  R_{cdjb}\gamma^j -ik\xi_c |\xi|^{k-2} w_{bc} \gamma^b + o({\bf x}),\\
         \mathfrak{w}^0_{k-2}&=-kw_{bcc}\gamma^b |\xi|^{k-2}-k(k-2)\xi_c\xi_d |\xi|^{k-4}w_{bcd}\gamma^b +o({\bf 1}).
     \end{align}
 \end{lemma}
     Next, we can formulate the following proposition:
\begin{proposition}\label{prop:f1comm}
    The only terms in $\W(\hat{u}[|D|^k, \hat{w}]|D|^{-n-k+2})$ that contain either derivatives of w or (combinations of ) Riemann tensor but are linear in $w$ read
     \begin{align}
     \nu_{n-1}\frac{k(n+k-2)}{2n}u_a\Bigl(-2\Tr(\I)w_{acc}-i \Tr\bigl([\gamma^a,\gamma^b]\{\gamma^c, B_0\}\bigr)w_{bc} \Bigr) +\ldots.
     \end{align}
\end{proposition}

 \begin{proof}
         Using lemma~\ref{symbolsDk}, we compute
         \begin{align}
             \sigma_{-n}([|D|^k,\hat{w}]|D|^{-n-k+2})=&\mathfrak w_{k-1}\mathfrak d_{(-n-k+2)-1}+\mathfrak w_{k-2}\mathfrak d_{-n-k+2}-i\partial_\xi^{a}\mathfrak w_{k-1}\partial_a\mathfrak d_{-n-k+2} \nonumber \\
             =&-\frac{ik(n+k-2)}{2}\xi_c\xi_d |\xi|^{-n-2} w_{bc} \gamma^b\{\gamma^d,B_0\}\nonumber \\
             &+i\frac{k}{2}|\xi|^{-n} \{\gamma^c, B_0\}w_{bc}\gamma^b
             \nonumber\\
             &+i\frac{k(k-2)}{2}|\xi|^{-n-2}\xi_c\xi_d \{\gamma^c, B_0\}w_{bd}\gamma^b-kw_{bcc}\gamma^b |\xi|^{-n} \nonumber\\
             &-k(k-2)\xi_c\xi_d |\xi|^{-n-2}w_{bcd}\gamma^b+ \ldots+o({\bf 1}).
         \end{align}
         Again, "$\ldots$" represents terms that do not affect this analysis.

         Next, we integrate this expression over the cosphere $|\xi|=1$,
         \begin{align}
            \frac{1}{\nu_{n-1}}\int_{|\xi|=1}\!\!\!\!\!\sigma_{-n}([|D|^k, \hat{w}]|D|^{-n-k+2})=&-\frac{ik(n+k-2)}{2n} w_{bc} \gamma^b\{\gamma^c,B_0\} \nonumber \\
             &+i\frac{k}{2} \{\gamma^c, B_0\}w_{bc}\gamma^b
             + i\frac{k(k-2)}{2n}\{\gamma^c, B_0\}w_{bc}\gamma^b\nonumber \\
             &-kw_{bcc}\gamma^b -\frac{k(k-2)}{n}w_{bcc}\gamma^b+ \ldots+o({\bf 1})\nonumber \\
             =&\frac{k(n+k-2)}{2n}\bigl[iw_{bc}\left(-\gamma^b\{\gamma^c,B_0\}+\{\gamma^c,B_0\}\gamma^b\right) \nonumber \\ &\hspace{2.3cm}-2w_{bcc}\gamma^b \bigr]+ \ldots+o({\bf 1}).
         \end{align}
Promptly, we have 
\begin{align}
    \Tr\Bigl[ \int_{|\xi|=1}\!\!\!\!\!\!\sigma_{-n}(\hat u[|D|^k, \hat{w}]|D|^{-n-k+2}) \Bigr]=&\Tr\Bigl[u_a\gamma^a\int_{|\xi|=1}\!\!\!\!\!\!\sigma_{-n}([|D|^k, \hat{w}]|D|^{-n-k+2})\Bigr]\nonumber\\
    =&\frac{k(n+k-2)}{2n}u_a\Tr\bigl[iw_{bc}\gamma^a\bigl(-\gamma^b\{\gamma^c,B_0\}\nonumber \\ &+\{\gamma^c,B_0\}\gamma^b\bigr)-2w_{bcc}\gamma^a\gamma^b \bigr]+ \ldots+o({\bf 1})\nonumber\\
    =&\frac{k(n+k-2)}{2n}u_a\Tr\bigl[iw_{bc}\left(-\gamma^a\gamma^b+\gamma^b\gamma^a\right)\{\gamma^c,B_0\}\nonumber \\ &-2w_{acc} \bigr]+ \ldots+o({\bf 1}).
\end{align}
 \end{proof}
Ultimately,, we formulate the main result of this subsection:
     \begin{proposition}\label{prop:f1}
         The density of $F_1(k)=\W\bigl(\hat u |D|^k \hat w |D|^{-n-k+2}\bigr)$ reads
\begin{equation}
\begin{aligned}
 F_1(k)= \nu_{n-1}\biggl[&\frac{2-n}{24}\Tr(\I)u_aw_a R -\frac{k(n+k-2)}{n}\Tr(\I)u_a w_{acc}\\
 &-i\frac{k(k+n-2)}{2n} u_a w_{bc} \Tr\bigl([\gamma^a,\gamma^b]\{\gamma^c,B_0\}\bigr)\biggr] 
 + (\ldots)
\end{aligned}
\end{equation}
     \end{proposition}
     \begin{proof}
         We use the fact that
\begin{equation}
    \W\bigl(\hat u |D|^k \hat w |D|^{-n-k+2}\bigr)=\W\bigl(\hat u [|D|^k, \hat w] |D|^{-n-k+2}\bigr)+\W\bigl(\hat u \hat w |D|^{-n+2}\bigr).
\end{equation}
         The first term was already computed in Proposition~\ref{prop:f1comm}. The second term cannot contain derivatives of $w$ so all interesting terms originate from $\W\bigl(\hat u \hat w |D_0|^{-n+2}\bigr)$, which is calculated in \cite[Section 4]{Dabrowski23} for even $n$, but is also true in odd dimensions (see discussion below Lemma~\ref{symbolsDk}).
     \end{proof}

\subsection{$F_2(k)$}
We now concentrate on the computation of $F_2(k)=\W\bigl(uD |D|^k w D|D|^{-n-k}\bigr)$. For that, we first need the following lemma:
\begin{lemma}
     The homogeneous symbols $\sigma(D|D|^k)=\mathfrak{e}_{k+1}+\mathfrak{e}_{k}+\mathfrak{e}_{k-1}$ read
     \begin{align}
         \mathfrak{e}_{k+1}=&\mathfrak{e}_{k+1}^0\\
         \mathfrak{e}_{k}=&\mathfrak{e}_{k}^0+\frac{k}{2}\gamma^c \xi_a\xi_c |\xi|^{k-2}\{\gamma^a, B\}+B|\xi|^k
         +o({\bf x}),\\
         \mathfrak{e}_{k-1}=&\mathfrak{e}_{k-1}^0-\gamma^c \xi_c \Bigl[\frac{k}{2}|\xi|^{k-2}\bigl(
         i\gamma^a B_a +B_0^2\bigr)\nonumber\\
         &+\frac{k(k-2)}{4}\xi_a\xi_b |\xi|^{k-4}\bigl(
         i\{\gamma^a, B_b\}+\frac{1}{2}\{\gamma^a, B_0\}\{\gamma^b, B_0\}\bigr)\Bigr]\nonumber\\
         &-\frac{k}{2}\xi_a |\xi|^{k-2}B_0 \{\gamma^a, B_0\}-\frac{ik}{2}\gamma^b \xi_a |\xi|^{k-2}\{\gamma^a, B_b\}+o({\bf 1}),
     \end{align}
     with
     \begin{align}
         \mathfrak{e}_{k+1}^0=&-\gamma^a \xi_a |\xi|^k +\frac{1}{6}\gamma^a R_{abcd} x^b x^{c} \xi_d |\xi|^k-\frac{k}{6}\gamma^e \xi_e \xi_a\xi_b R_{acbd}x^c x^d |\xi|^{k-2}+o({\bf x^2}),\\
         \mathfrak{e}_{k}^0=&k\gamma^c \xi_a\xi_c |\xi|^{k-2}x^b\bigl(-\frac{i}{3}\mathrm{Ric}_{ab} -\frac{i}{8}R_{abjk} \gamma^j\gamma^k\bigr)\nonumber \\
         &-\frac{i}{4}\mathrm{Ric}_{ab}\gamma^a x^b |\xi|^k +\frac{ik}{3}\gamma^c R_{acbd}\xi_a\xi_b |\xi|^{k-2} x^d +o({\bf x}),\\
         \mathfrak{e}_{k-1}^0=&-\gamma^c \xi_c \Bigl(\frac{k}{8}|\xi|^{k-2}R+\frac{k(k-2)}{12}\xi_a\xi_b |\xi|^{k-4}\mathrm{Ric}_{ab}\Bigr)\nonumber\\
         &-k\gamma^b \xi_a |\xi|^{k-2}\bigl(
         \frac{1}{3}\mathrm{Ric}_{ab}+\frac{1}{8}R_{abjk}\gamma^j \gamma^k\bigr) +o({\bf 1}).
     \end{align}
\end{lemma}

\begin{proof}
        The result follows from combining the symbols of $D=D_0+B$ (with $D_0$ as in \eqref{Dirac}) and $|D|^k$ (see Lemma~\ref{symbolsDk}).
    
\end{proof}
 \begin{proposition}
     The only terms in $F_2(k)=\W(\hat{u}D|D|^k\hat{w}D|D|^{-n-k})$ that contain either derivatives of w or (combinations of ) Riemann tensor but are linear in $w$ read
     \begin{align}
         \frac{F_2(k)}{\nu_{n-1}}= & \, u_aw_b\Tr(\I)\Bigl(\frac16 G_{ab}+\frac{n-2}{24}R\delta_{ab}\Bigr)+\frac{k(n+k)}{n(n+2)}u_a (n w_{abb}-4 w_{bba})\nonumber \\
         &+i\frac{k(n+k)+n}{2n}u_aw_{bc}\Tr\bigl([\gamma^a,\gamma^b]\{\gamma^c, B_0\}\bigr) +\ldots .
     \end{align}
 \end{proposition}
 \begin{proof}
     We begin by noting that
     \begin{align}
         \sigma_{-n}\bigl(\hat{u}D|D|^k \hat{w}D|D|^{-n-k}\bigr)= &\hat{u}\Bigl[\mathfrak{e}_{k+1}\hat{w}\mathfrak{e}_{-n-k-1}+\mathfrak{e}_k \hat{w}\mathfrak{e}_{-n-k} +\mathfrak{e}_{k-1}\hat{w}\mathfrak{e}_{-n-k+1}\nonumber\\
         &-i\partial^c_\xi \mathfrak{e}_{k+1}\partial_c \bigl(\hat{w}\mathfrak{e}_{-n-k} \bigr)-i\partial^c_\xi \mathfrak{e}_k \partial_c \bigl(\hat{w}\mathfrak{e}_{-n-k+1}\bigr)\nonumber\\
         &-\frac{1}{2}\partial_\xi^{cd}\mathfrak{e}_{k+1}\partial_{cd}\bigl(\hat{w}\mathfrak{e}_{-n-k+1}\bigr)
         \Bigr]\nonumber \\         =&\hat{u}\Bigl[\mathfrak{e}_{k+1}^0\hat{w}\mathfrak{e}_{-n-k-1}^0+\mathfrak{e}_k^0 \hat{w}\mathfrak{e}_{-n-k}^0 +\mathfrak{e}_{k-1}^0\hat{w}\mathfrak{e}_{-n-k+1}^0\nonumber\\
         &-i\partial^c_\xi \mathfrak{e}_{k+1}^0\hat{w}\partial_c \mathfrak{e}_{-n-k}^0-i\partial^c_\xi \mathfrak{e}_k^0 \hat{w}\partial_c \mathfrak{e}_{-n-k+1}^0\nonumber\\
         &-\frac{1}{2}\partial_\xi^{cd}\mathfrak{e}_{k+1}^0\hat{w}\partial_{cd}\mathfrak{e}_{-n-k+1}^0\nonumber \\
         &-i\partial^c_\xi \mathfrak{e}_{k+1}\partial_c \hat{w}\mathfrak{e}_{-n-k} -i\partial^c_\xi \mathfrak{e}_k \partial_c \hat{w}\mathfrak{e}_{-n-k+1}\nonumber\\
         &-\partial_\xi^{cd}\mathfrak{e}_{k+1}\partial_{c}\hat{w}\partial_d\mathfrak{e}_{-n-k+1}-\frac{1}{2}\partial_\xi^{cd}\mathfrak{e}_{k+1}\partial_{cd}\hat{w}\mathfrak{e}_{-n-k+1}
         \Bigr]\nonumber\\&+(\ldots),
     \end{align}
     where "$\ldots$" stands for irrelevant terms in our analysis (i.e. with $B$ but without derivatives of $w$).
     
     By a direct inspection, we see that only the very last term could lead to the presence of second derivatives,    
     \begin{align}
         \Tr\int_{|\xi|=1}\!\!\!\!&\hat{u}\Bigl[-\frac{1}{2}\partial_\xi^{cd}\mathfrak{e}_{k+1} \partial_{cd}\hat{w}\mathfrak{e}_{-n-k+1}\Bigr]\nonumber \\ =& \Tr\int_{|\xi|=1}u_a\gamma^a\frac{k}{2}\bigl(\gamma^c \xi_d +\gamma^d \xi_c +\delta_{cd}\xi_j\gamma^j +(k-2)\gamma^j \xi_j\xi_c\xi_d\bigr)\bigl(-2w_{bcd}\gamma^b \gamma^e \xi_e\bigr)\nonumber\\
         =& -\nu_{n-1}\frac{k}{n}u_a\Tr\Bigl[\gamma^a\gamma^c\gamma^b\gamma^d w_{bcd}+\gamma^a\gamma^d\gamma^b\gamma^c+\gamma^a\gamma^e\gamma^b\gamma^e w_{bcc}\nonumber \\
         &+\frac{k-2}{n+2}\bigl(\gamma^a\gamma^c\gamma^b\gamma^d w_{bcd}+\gamma^a\gamma^d\gamma^b\gamma^c w_{bcd}+\gamma^a\gamma^e\gamma^b\gamma^e w_{bcc}\bigr)\Bigr]\nonumber \\
         =&-\nu_{n-1}\frac{k(n+k)}{n(n+2)}u_a\Tr\Bigl[\gamma^a\gamma^c\gamma^b\gamma^d w_{bcd}+\gamma^a\gamma^d\gamma^b\gamma^c+\gamma^a\gamma^e\gamma^b\gamma^e w_{bcc}\Bigr]\nonumber \\
         =&-\nu_{n-1}\frac{k(n+k)}{n(n+2)}\Tr(\I)u_a\Bigl[w_{acc}\bigl(-1-1+2-n\bigr)+w_{bba}\bigl(2+2\bigr)\Bigr].
     \end{align}

     Next, only three terms involve derivatives of $w$ (and the third one vanishes in $x=0$). They read:
     \begin{align}
         \Tr\int_{|\xi|=1}&\hat{u}\Bigl[-i\partial^c_\xi \mathfrak{e}_{k+1}\partial_c \hat{w}\mathfrak{e}_{-n-k} -i\partial^c_\xi \mathfrak{e}_k \partial_c \hat{w}\mathfrak{e}_{-n-k+1}
         -\partial_\xi^{cd}\mathfrak{e}_{k+1}\partial_{c}\hat{w}\partial_d\mathfrak{e}_{-n-k+1}\Bigr]\nonumber\\
         =&\Tr\int_{|\xi|=1}u_a\gamma^a\Bigl[i\bigl(\gamma^c+k\gamma^d\xi_c\xi_d\bigr)w_{bc}\gamma^b\bigl(\frac{-n-k}{2}\gamma^j \xi_j\xi_k\{\gamma^k, B_0\}+B_0\bigr)  \nonumber \\ &+ikB_0\xi_cw_{bc}\gamma^b\gamma^d\xi_d
         +\frac{ik}{2}\bigl(\gamma^j\delta_{cd}\xi_j+\gamma^c\xi_d+(k-2)\gamma^j\xi_c\xi_d\xi_j\bigr)\{\gamma^d,B_0\}w_{bc}\gamma^b\gamma^k\xi_k\Bigr]\nonumber \\
         =&i\nu_{n-1}u_aw_bc\Tr\Bigl\{\gamma^a\gamma^c\gamma^b B_0-\frac{n+k}{2n}\gamma^a\gamma^c\gamma^b\gamma^d\{\gamma^d,B_0\}+\frac{k}{n}\gamma^a\Bigl[ \gamma^c\gamma^b B_0\nonumber \\
         &-\frac{n+k}{2(n+2)}\bigl(\gamma^c\gamma^b\gamma^j\{\gamma^j,B_0\}+ \gamma^d\gamma^b\gamma^c\{\gamma^d,B_0\}+ \gamma^d\gamma^b\gamma^d\{\gamma^c,B_0\}
         \bigr)\nonumber \\
         &+B_0\gamma^b\gamma^c+\frac12\gamma^j\{\gamma^c,B_0\}\gamma^b\gamma^j+\frac12\gamma^c\{\gamma^d,B_0\}\gamma^b\gamma^d\nonumber \\
         &+\frac{k-2}{2(n+2)}\bigl(\gamma^j\{\gamma^c,B_0\}\gamma^b\gamma^j+\gamma^c\{\gamma^d,B_0\}\gamma^b\gamma^d+\gamma^d\{\gamma^d,B_0\}\gamma^b\gamma^c
         \bigr)\Bigr]\Bigr\}
     \end{align}
     We can simplify this expression by collecting similar terms, using  cyclicity of the trace and anticommutation rules of $\gamma$'s to get rid of $\gamma$ matrices with the same indices:
     \begin{align}
         i\nu_{n-1}u_aw_bc\Tr\Biggl\{&\!-\frac{2k(k+n)}{n(n+2)}\gamma^a\gamma^c\gamma^b B_0-\frac{(n+k)(n+k+2)}{n(n+2)}\bigl(-\gamma^a\gamma^b\gamma^cB_0+\gamma^c\gamma^b\gamma^aB_0\bigr)\nonumber \\
         &+\frac{k}{n}\Bigl[
         -\frac{(n+k)(2-n)}{2(n+2)} \gamma^a\gamma^b\{\gamma^c,B_0\}+\frac{(n+k)(2-n)}{2(n+2)}\gamma^b\gamma^a\{\gamma^c,B_0\} \nonumber
         \\         &+\frac{2(n+k)}{n+2}\gamma^b\gamma^c\gamma^aB_0+\frac{k-2}{n+2}\bigl(-\gamma^b\gamma^a\gamma^cB_0+\gamma^c\gamma^a\gamma^bB_0
         \bigr)\Bigr]\Biggr\}.
     \end{align}
     We then want to express every term as a part with $\{\gamma^c,B_0\}+(\ldots)$, leading to
     \begin{align}
         &\frac{i}{n(n+2)}\nu_{n-1}u_aw_bc\Tr\Biggl[k(k+n)\Bigl(\gamma^a\gamma^b\{\gamma^c,B_0\}-2\delta_{ac}\gamma^bB_0-2\delta_{bc}\gamma^aB_0\Bigl)\nonumber\\
         &+\frac12(n+k)(n+k+2)\Bigl(\gamma^a \gamma^b\{\gamma^c,B_0\}-\gamma^b \gamma^a\{\gamma^c,B_0\} \Bigr)\nonumber\\
         &+\frac12k(n+k)(2-n)\Bigl( -\gamma^a\gamma^b\{\gamma^c,B_0\}+\gamma^b\gamma^a\{\gamma^c,B_0\}\Bigr)\nonumber\\
         &+k(k+n)\Bigl(-\gamma^b \gamma^a\{\gamma^c,B_0\}+2\delta_{ac}\gamma^bB_0+2\delta_{bc}\gamma^aB_0\Bigl)\nonumber\\
         &+\frac12k(k-2)\Bigl(\gamma^a \gamma^b\{\gamma^c,B_0\}-\gamma^b \gamma^a\{\gamma^c,B_0\} \Bigr)
         \Biggr]\nonumber\\
         =&\frac{i}{n(n+2)}\nu_{n-1}u_aw_bc\Tr\Biggl([\gamma^a,\gamma^b]\{\gamma^c,B_0\}\Bigl[k(n+k)\nonumber\\
         &+ \frac12(n+k)(n+k+2)+\frac12k(n+k)(n-2)+\frac12k(k-2)\Bigr]
         \Biggr)\nonumber\\
         =&i\frac{k(n+k)+n}{2n}\nu_{n-1}u_aw_bc\Tr\Bigl([\gamma^a,\gamma^b]\{\gamma^c,B_0\}\Bigl).
     \end{align}
     We also need to calculate terms without derivatives of $w$. The procedure is similar to the proof of Proposition~\ref{prop:f1}. We split the part without $B$ into:
    \begin{align}
         \W\bigl(\hat u D_0|D_0|^k \hat w D_0|D_0|^{-n-k}\bigr)= & \W\bigl(\hat u D_0 [|D_0|^k, \hat w] D_0|D_0|^{-n-k}\bigr)\nonumber\\
         &+\W\bigl(\hat u D_0 \hat w D_0|D_0|^{-n}\bigr).
\end{align}
The second term is calculated in \cite[Section 4]{Dabrowski23} for even $n$, but again, the result is valid also for odd $n$ (see discussion below Lemma~\ref{symbolsDk}). On the other hand, to compute the first term we need the symbol of $\hat u D_0 [|D_0|^k, \hat w]$ or rather the part of it without derivatives of $w$, as they are already handled. By Lemma~\ref{lemm:WD40} we see that the symbol of $[|D|^k, \hat w]$ has only one term without $B$ or derivatives of $w$:
        \begin{equation}
         \mathfrak{w}^0_{k-1}=\frac{ik}{2}\xi_c|\xi|^{k-2}w_bx^d R_{cdjb}\gamma^j+\ldots+o({\bf x}),
\end{equation}   \
and it is easy to check that
         \begin{align}
         \sigma\bigl(\hat u D_0 [|D_0|^k, \hat w]\bigr)
         =-\frac{k}{2}\xi_c|\xi|^{k-2}u_aw_b\gamma^a \gamma^dR_{cdjb}\gamma^j+\ldots+o({\bf 1}).
         \end{align}
We then have
         \begin{align}
             \Tr&\int_{|\xi|=1}\!\!\!\!\sigma_{-n}\bigl(\hat u D_0 [|D_0|^k, \hat w]D_0|D_0|^{-n-k}\bigr)\nonumber\\
             =&\Tr\int_{|\xi|=1}\!\!\!\!-\frac{k}{2}\xi_c|\xi|^{k-2}u_aw_b\gamma^a \gamma^dR_{cdjb}\gamma^j\bigl(-\gamma^k\xi_k|\xi|^{-n-k}\bigr)\nonumber\\
             =&\frac{k}{2n}\nu_{n-1}u_aw_bR_{cdjb}\Tr\bigl(\gamma^a\gamma^d\gamma^j\gamma^c\bigr)=\frac{k}{2n}\nu_{n-1}\Tr\bigl(\I\bigr)u_aw_b\bigl(\mathrm{Ric}_{ab}-\mathrm{Ric}_{ab}\bigr)=0,
         \end{align}
         and finally
         \begin{align}
             \W\bigl(\hat u D_0|D_0|^k \hat w D_0|D_0|^{-n-k}\bigr)&=\W\bigl(\hat u D_0 \hat w D_0|D_0|^{-n}\bigr)+\ldots\nonumber \\
             &=\nu_{n-1}u_aw_b\Tr(\I)\Bigl(\frac16 G_{ab}+\frac{n-2}{24}R\delta_{ab}\Bigr)+\ldots.
         \end{align}
 \end{proof}
 \bibliography{apssamp}